\documentclass[runningheads]{llncs}

\usepackage{amssymb}
\usepackage{graphicx}
\usepackage[nomessages]{fp}

\usepackage{amsmath}
\usepackage{tikz}
\usepackage{pgfplots}
\usepgfplotslibrary{fillbetween}
\usepackage{xcolor}

\usepackage{hyperref}

\newcommand{\grsmn}[3]{\mathcal{G}_{#1}\left(#2,#3\right)}
\DeclareMathOperator{\Span}{span}
\DeclareMathOperator{\matricize}{mat}
\DeclareMathOperator{\rank}{rank}

\newcommand{\bF}{\mathbb{F}}

\newcommand{\cQ}{\mathcal{Q}}
\newcommand{\cR}{\mathcal{R}}

\newcommand{\cU}{\mathcal{U}}

\newcommand{\cW}{\mathcal{W}}

\newcommand{\bolda}{\mathbf{a}}
\newcommand{\boldb}{\mathbf{b}}
\newcommand{\boldc}{\mathbf{c}}
\newcommand{\boldd}{\mathbf{d}}

\newcommand{\boldn}{\mathbf{n}}

\newcommand{\boldu}{\mathbf{u}}
\newcommand{\boldv}{\mathbf{v}}
\newcommand{\boldw}{\mathbf{w}}

\newcommand{\boldy}{\mathbf{y}}
\newcommand{\boldz}{\mathbf{z}}

\newcommand{\boldA}{\mathbf{A}}
\newcommand{\boldB}{\mathbf{B}}
\newcommand{\boldC}{\mathbf{C}}

\newcommand{\boldE}{\mathbf{E}}

\newcommand{\boldI}{\mathbf{I}}

\newcommand{\boldM}{\mathbf{M}}
\newcommand{\boldN}{\mathbf{N}}

\newcommand{\boldbeta}{\boldsymbol{\beta}}

\newcommand{\boldmu}{\boldsymbol{\mu}}
\newcommand{\boldnu}{\boldsymbol{\nu}}

\newcommand{\boldupsilon}{\boldsymbol{\upsilon}}

\newcommand{\boldOmega}{\mathbf{\Omega}}

\newcommand{\boldGamma}{\mathbf{\Gamma}}

\newif\ifePrint
\ePrinttrue

\newcommand\blfootnote[1]{%
 \begingroup
 \renewcommand\thefootnote{}\footnote{#1}%
 \addtocounter{footnote}{-1}%
 \endgroup
}

\newtheorem{construction}{Construction}

\begin{document}
\title{Multivariate Public Key Cryptosystem\\from Sidon Spaces}
\author{Netanel Raviv\inst{1} \and Ben Langton\inst{2} \and Itzhak Tamo\inst{3}}
\authorrunning{N.~Raviv, B.~Langton, and I.~Tamo}

\institute{
	Department of Computer Science and Engineering, Washington University in St.~Louis, St.~Louis, MO 63103, USA, \email{netanel.raviv@wustl.edu} \and 
	Department of Mathematics, Harvey Mudd College, Claremont, CA 91711, USA, \email{blangton@g.hmc.edu} \and 
	Department of Electrical Engineering--Systems, Tel-Aviv University, Tel-Aviv, Israel, \email{zactamo@gmail.com}
}

\maketitle
\begin{abstract}
A Sidon space is a subspace of an extension field over a base field in which the product of any two elements can be factored uniquely, up to constants. This paper proposes a new a public-key cryptosystem of the multivariate type which is based on Sidon spaces, and has the potential to remain secure even if quantum supremacy is attained. This system, whose security relies on the hardness of the well-known MinRank problem, is shown to be resilient to several straightforward algebraic attacks. In particular, it is proved that the two popular attacks on the MinRank problem, the kernel attack and the minor attack, succeed only with exponentially small probability. The system is implemented in software, and its hardness is demonstrated experimentally.
\end{abstract}

\section{Introduction}
\blfootnote{© IACR 2016. This article is the final version submitted by the authors to the IACR and to Springer-Verlag on March 2nd 2021. The final version is also avaialble at \cite{pkc}.}
Public key cryptosystems (PKCs), such as RSA, are essential in many communication scenarios. However, most number-theoretic PKCs are prone to quantum attacks, and will become obsolete once quantum supremacy is attained. Hence, it is important to devise PKCs whose hardness relies on problems that are hard to solve even with quantum computers at hand. One such problem that has gained increasing attention lately is solving a system of multivariate polynomial (usually quadratic) equations~\cite{HFE,ABC}, which is NP-hard in general~\cite{MinRankCryptanalysis}. PKCs whose hardness relies on solving multivariate polynomial equations are called \textit{Multivariate Public Key Cryptosystems} (MPKCs).

Nearly all MPKCs in the literature were either cryptanalyzed or their efficiency was proved to be limited. Recently, the so-called ABC cryptosystem~\cite{ABC}, that relies on simple matrix multiplication as the encryption scheme, seems to have been broken by~\cite{ABCcryptanalysis}. Earlier works include the Hidden Field Equations (HFE) cryptosystem~\cite{HFE} that has been broken by a MinRank attack~\cite{HFEcryptanalysis}, and the TTM scheme~\cite{TTM} that experienced a similar fate~\cite{TTMcryptanalysis}. In addition, some variants of the HFE scheme, such as ZHFE~\cite{ZHFE} and HFE\textsuperscript{-}~\cite{HFE-} were also successfully attacked by~\cite{ZHFEcryptanalysis} and~\cite{HFE-cryptanalysis}. The MPKC of~\cite{Schulman} was also broken by a MinRank attack. Additional candidates include the oil-and-vinegar scheme~\cite{UOV}, Rainbow~\cite{Rainbow}, and Gui~\cite{Gui}. In light of these recent advances, in this paper we propose a new MPKC which seems to be inherently robust against several natural attacks. This MPKC is based on a newly defined algebraic concept called \textit{Sidon spaces}.

Let~$\bF_q$ denote a finite field with~$q$ elements, and let~$\bF_q^*\triangleq \bF_q\setminus\{0\}$. For an integer~$n$ let~$\bF_{q^n}$ be the algebraic extension of degree~$n$ of~$\bF_q$, and let~$[n]\triangleq \{1,2,\ldots,n\}$. Simply put, a Sidon space~$V$ is a subspace of~$\bF_{q^n}$ over~$\bF_q$ such that the product of any two nonzero elements of~$V$ has a unique factorization over~$V$, up to a constant multiplier from~$\bF_q$.
Sidon spaces were recently defined in~\cite{Vospers} as a tool for studying certain multiplicative properties of subspaces, and their application to error correction in network coding, alongside several explicit constructions that are employed herein, were studied in~\cite{SidonSpaces}. 

In this paper we suggest the \textit{Sidon Cryptosystem}, an MPKC based on Sidon spaces. In a nutshell, this cryptosystem enables the sender to transmit the product of two elements in a secret Sidon space~$V$, without knowing its structure. The receiver uses the structure of~$V$ in order to factor the given product and obtain the plaintext efficiently. A malicious attacker, however, cannot extract the plaintext from the product due to insufficient knowledge about~$V$. The suggested Sidon cryptosystem is based on a specific optimal construction of a Sidon space from~\cite{SidonSpaces}, and yet, other Sidon spaces with comparable parameters can be used similarly. The security of the suggested system relies on the hardness of solving multivariate polynomial equations, and the hardness of the MinRank problem. 

In the MinRank problem, which is believed to be hard even for quantum computers, one must find a low-rank target matrix in the linear span of matrices that are given as input; it arises in settings where one solves a quadratic system of equations via linearization. Cryptographic systems that are based on the MinRank attack are often broken by either of two attacks, the minor attack and the kernel attack (also known as the Kipnis-Shamir attack)~\cite{ComputingLoci}. In the minor attack, one formulates an equation system by setting all small minors of the target matrix to zero, and solves the resulting system (usually) by linearization. The kernel attack exploits the fact that vectors in the kernel of the target matrix give rise to linear equations in the coefficients of its combination; successfully guessing sufficiently many of those will break the system. 

In the sequel we analyze the Sidon cryptosystem in the face of these two attacks, and both are proved to succeed only with exponentially small probability. We additionally analyze attacks that are specific to the Sidon cryptosystem and are not in either of those forms, and show that these require solving polynomial equations outside the range of feasibility. We emphasize that unfortunately, a rigorous proof of hardness that does not rely on a particular attack structure is yet to be found.

This paper is organized as follows. The definition of a Sidon space, alongside relevant constructions and their efficient factorization algorithm, are given in Section~\ref{section:Preliminaries}. The details of the Sidon cryptosystem are given in Section~\ref{section:TheSidonCryptosystem}, and its efficiency is discussed. MinRank attacks, namely, the kernel attack and the minor attack, are discussed in Section~\ref{section:MinRankAttacks}. Attacks which are specifically designed for the Sidon cryptosystem are discussed in Section~\ref{section:OtherAttacks}. Finally, experimental results are reported in Section~\ref{section:experiments}, and concluding remarks are given in Section~\ref{section:Discussion}. 

We adopt the following notational conventions. Scalars are denoted by~$a,b,\ldots$ or~$\alpha,\beta,\ldots$; matrices by~$\boldA,\boldB\ldots$; sets by~$\cU,\cW,\ldots$; linear subspaces and polynomials by~$V,U,\ldots$; and vectors, all of which are row vectors, by~$\boldv,\boldnu,\ldots$. 

\section{Preliminaries}\label{section:Preliminaries}
For integers~$k$ and~$n$ let $\grsmn{q}{n}{k}$ be the set of all~$k$-dimensional subspaces of~$\bF_{q^n}$ over~$\bF_q$. \emph{Sidon spaces} were recently defined in~\cite{Vospers} as a tool for studying certain multiplicative properties of subspaces. 
As noted in~\cite{Vospers}, the term ``Sidon space'' draws its inspiration from a \emph{Sidon set}. A set of integers is called a Sidon set if the sums of any two (possibly identical) elements in it are distinct;
thus, Sidon spaces may be seen as a multiplicative and linear variant of Sidon sets.
In the following definition, for~$a\in\bF_{q^n}$, let $a\bF_{q}\triangleq\{\lambda a \vert \lambda\in\bF_q \}$, which is the subspace over~$\bF_q$ spanned by~$a$.

\begin{definition}[{\cite[Sec.~1]{Vospers}}]\label{definition:SidonSpace}
	A subspace~$V\in\grsmn{q}{n}{k}$ is called a Sidon space if for all nonzero~$a,b,c,d\in V$, if~$ab=cd$ then~$\{a\bF_q,b\bF_q \}=\{c\bF_q,d\bF_q \}$.
\end{definition}

It is shown in~\cite[Thm.~18]{Vospers} and in~\cite[Prop.~3]{SidonSpaces} that if~$V\in\grsmn{q}{n}{k}$ is a Sidon space, then
\begin{align}\label{equation:SidonBounds}
2k\le \dim(V^2)\le \binom{k+1}{2},
\end{align}
where~$V^2\triangleq \Span_{\bF_q}\{u\cdot v\vert u,v\in V\}$, and consequently it follows that~$k\le n/2$. Sidon spaces which attain the upper bound are called \textit{max-span} Sidon spaces, and are rather easy to construct; it is an easy exercise to verify that~$\Span_{\bF_q}\{ \delta^{n_i} \}_{i=1}^k$ is a max-span Sidon space in~$\bF_{q^n}$ for~$n>2k^2(1+o_k(1))$, where~$\delta$ is a primitive element of~$\bF_{q^n}$, and~$\{n_1,\ldots,n_k\}\subseteq [\lfloor n/2 \rfloor]$ is an optimal (i.e., largest) Sidon set~\cite{SidonSets}. Sidon spaces which attain the lower bound in~\eqref{equation:SidonBounds} are called \textit{min-span} Sidon spaces, and are paramount to our work; it will be shown in the sequel that having~$k=\Theta(n)$ is essential to the security of the system. The cryptosystem which is given in this paper employs a min-span Sidon space whose construction is given in the remainder of this section.

Motivated by applications in network coding, constructions of Sidon spaces in several parameter regimes were suggested in~\cite{SidonSpaces}. In particular, the following construction provides a Sidon space~$V\in\grsmn{q}{rk}{k}$ for any~$k$ and any~$r\ge 3$. A slightly more involved variant of this construction is shown in the sequel to provide a Sidon space in~$\grsmn{q}{2k}{k}$, which will be used in our cryptosystem.

\begin{construction} \label{construction:n>=3k} \cite[Const.~11]{SidonSpaces}
	For integers~$r\ge 3$, $k$, and~$q$ a prime power, let~$\gamma\in\bF_{q^{rk}}^*$ be a root of an irreducible polynomial of degree~$r$ over~$\bF_{q^k}$. Then, $V\triangleq\{u+u^q\gamma \vert u\in\bF_{q^k} \}$ is a Sidon space in~$\grsmn{q}{rk}{k}$.
\end{construction}

By choosing the element~$\gamma$ judiciously, a similar construction provides a Sidon space in~$\grsmn{q}{2k}{k}$ for any~$q\ge 3$ as follows.  
For any given nonnegative integer~$k$, let~$\cW_{q-1}\triangleq \{ u^{q-1}\vert u\in\bF_{q^k} \}$ and~$\overline{\cW}_{q-1}\triangleq \bF_{q^k}\setminus \cW_{q-1}$. The next construction requires an element~$\gamma\in\bF_{q^{2k}}$ that is a root of an irreducible quadratic polynomial~$x^2+bx+c$ over~$\bF_{q^k}$, where~$c\in\overline{\cW}_{q-1}$. According to~\cite[Lemma~13]{SidonSpaces}, for any~$c\in\overline{\cW}_{q-1}$ there exist many~$b$'s in~$\bF_{q^k}$ such that~$x^2+bx+c$ is irreducible over~$\bF_{q^k}$, and hence such~$\gamma$ elements abound\footnote{Since~$|\overline{\cW}_{q-1}|=q^k-\frac{q^k-1}{q-1}-1$, a crude lower bound for the number of such elements~$\gamma$ is~$\approx\frac{q-2}{q-1}\cdot q^{k}$.}.

\begin{construction}\label{construction:n=2k} \cite[Const.~15]{SidonSpaces}
	For a prime power~$q\ge 3$ and a positive integer~$k$, let~$n=2k$, and let~$\gamma\in\bF_{q^n}^*$ be a root of an irreducible polynomial~$x^2+bx+c$ over~$\bF_{q^k}$ with~$c\in\overline{\cW}_{q-1}$. The subspace~$V\triangleq \{ u+u^q\gamma\vert u\in\bF_{q^k} \}$ is a Sidon space in~$\grsmn{q}{2k}{k}$.
\end{construction}

The Sidon space~$V$ of Construction~\ref{construction:n=2k} will be used in the sequel to devise an MPKC. 
This subspace admits the following efficient algorithm~\cite[Thm.~16]{SidonSpaces} that for every nonzero~$a$ and~$b$ in~$V$, factors~$ab$ to~$a$ and~$b$ up to constant factors from~$\bF_q$; note that since~$ab=(\frac{1}{\lambda}a)(\lambda b)$ for any~$a,b\in V$ and any~$\lambda\in\bF_q^*$, this algorithm is capable of identifying~$a$ and~$b$ only up to a multiplicative factor in~$\bF_q$. 

Given~$ab$, denote~$a=u+u^q\gamma$ for some nonzero~$u\in\bF_{q^k}$ and $b=v+v^q\gamma$ for some nonzero~$v\in\bF_{q^k}$. Notice that since~$\gamma$ is a root of~$x^2+bx+c$ it follows that
\begin{align*}
ab&=(u+u^q\gamma)(v+v^q\gamma)\\
&=(uv-(uv)^qc)+(uv^q+u^qv-b(uv)^q)\gamma\;,
\end{align*}
and since~$\{ 1,\gamma \}$ is a basis of~$\bF_{q^n}$ over~$\bF_{q^k}$, it follows that one can obtain the values of~$q_0\triangleq uv-(uv)^qc$ and~$q_1\triangleq uv^q+u^qv-b(uv)^q$ by representing~$ab$ over this basis.

Since~$c\in\overline{\cW}_{q-1}$, it follows that the linearized polynomial~$T(x)=x-cx^q$ is invertible on~$\bF_{q^k}$. Hence, it is possible to extract~$uv$ from~$q_0=T(uv)$ by applying~$T^{-1}$.

Knowing~$uv$, extracting~$uv^q+vu^q$ from~$q_1$ is possible by adding~$b(uv)^q$. Therefore, the polynomial $uv+(uv^q+u^qv)x+(uv)^qx^2$ can be assembled, and its respective roots~$-1/u^{q-1}$ and~$-1/v^{q-1}$ can be found. Since these roots determine~$u\bF_q$ and~$v\bF_q$ uniquely, it follows that~$a$ and~$b$ are identified up to order and up to a multiplicative factor in~$\bF_q$.

\section{The Sidon Cryptosystem}\label{section:TheSidonCryptosystem}
In general, for a Sidon space~$V\in \grsmn{q}{n}{k}$ and~$a,b\in V$, factoring~$ab$ to~$a$ and~$b$ requires knowledge about the structure of~$V$, as can be seen from the factoring algorithm suggested above. This intuition leads to the following MPKC, called the \textit{Sidon Cryptosystem}. The crux of devising this system is enabling Bob to encrypt his message into a product~$ab$, without the need to know the precise construction of~$V$ by Alice. This is done by exploiting the bilinear nature of multiplication in finite field extensions.

To this end, we introduce the notion of a \textit{multiplication table} of a subspace. For a given vector $\boldv=(v_1,\ldots,v_k)\in\bF_{q^n}^k$ let~$\boldM(\boldv)\triangleq \boldv^\intercal \boldv$. For an ordered basis~$\boldb=( b_1,b_2,\ldots,b_{n} )$ of~$\bF_{q^n}$ over~$\bF_q$, express~$\boldM(\boldv)$ as a linear combination of matrices over~$\bF_q$, i.e.,
\begin{align*}
	\boldM(\boldv)&=b_1\boldM^{(1)}+b_2\boldM^{(2)}+\ldots+b_{n}\boldM^{(n)},\text{ and let}\\
	\boldM(\boldv,\boldb)&\triangleq (\boldM^{(1)},\boldM^{(2)},\ldots,\boldM^{(n)}). 
\end{align*}
The matrix~$\boldM(\boldv)$ is called \textit{the multiplication table} of~$\boldv$, and the entries of~$\boldM(\boldv,\boldb)$ are called the \textit{coefficient matrices} of~$\boldv$ with respect to~$\boldb$. This notion will be of interest when $\{v_1,\ldots,v_k\}$ is a basis of a Sidon space.

The following cryptosystem relies on choosing a random Sidon space by Construction~\ref{construction:n=2k} (which amounts to randomly choosing a proper~$\gamma$), fixing an arbitrary ordered basis~$\boldnu=(\nu_1,\ldots,\nu_k )$ of~$V$, and interpreting the product of two elements~$a\triangleq \sum_{i=1}^ka_i\nu_i$ and $b\triangleq \sum_{i=1}^{k}b_i\nu_i$ in~$V$ as the bilinear form~$\bolda \boldM(\boldnu)\boldb^\intercal$, where $\bolda=(a_1,\ldots,a_k)\in\bF_q^k$ and~$\boldb=(b_1,\ldots,b_k)\in\bF_q^k$. Even though the suggested cryptosystem relies on Construction~\ref{construction:n=2k}, any Sidon space for which an efficient factorization algorithm exists may be used similarly. A remark about the required ratio~$k/n$ is given shortly.

To describe the message set in the following cryptosystem, let~$\sim$ be an equivalence relation on~$\bF_q^{k\times k}$ such that~$\boldA\sim \boldB$ if~$\boldA=\boldB^\intercal$, for any~$\boldA,\boldB\in\bF_q^{k\times k}$. Further, let~$\cQ_k$ be the set of $k\times k$ rank one matrices over~$\bF_q$, modulo the equivalence relation~$\sim$. That is, $\cQ_k$ is a set of equivalence classes, each of which contains either one symmetric matrix of rank one, or two non-symmetric matrices of rank one, where one is the transpose of the other. 
\ifePrint
It is shown in Lemma~\ref{lemma:Qk} in \nameref{appendix:omittedProof} that $|\cQ_k|=\frac{(q^k-1)(q^k-q)}{2(q-1)}+q^k-1$. 
\fi
In what follows, Alice chooses a random Sidon space~$V$ by Construction~\ref{construction:n=2k}, and publishes its coefficient matrices according to a random basis of~$V$ and a random basis of~$\bF_{q^n}$. Bob then sends Alice an encrypted message by exploiting the bilinear nature of multiplication in field extensions.

\begin{description}
	\item[Parameters:] An integer~$k$ and a field size~$q\ge 3$.
	\item[Private key:] Alice chooses
	\begin{enumerate}
		\item \textbf{A random representation of~$\bF_{q^n}$ over~$\bF_q$:} i.e., a polynomial~$P_A(x)\in\bF_q[x]$ of degree~$n=2k$ which is irreducible over~$\bF_{q}$.
		\item \textbf{A random Sidon space by Construction~\ref{construction:n=2k}:} i.e., a random element~$c\in\overline{\cW}_{q-1}$ and an element~$b\in\bF_{q^k}$ such that~$P_{b,c}(x)\triangleq x^2+bx+c$ is irreducible over~$\bF_{q^k}$, a~$\gamma\in\bF_{q^n}^*$ such that~$P_{b,c}(\gamma)=0$; this~$\gamma$ defines the Sidon space~$V\triangleq\{ u+u^q\gamma\vert u\in\bF_{q^k} \}$.
		\item \textbf{A random \textit{ordered} basis~$\boldnu=( \nu_1,\ldots,\nu_k )$ of~$V$:} which is equivalent to choosing a random invertible~$k\times k$ matrix over~$\bF_q$.
		\item \textbf{A random \textit{ordered} basis~$\boldbeta = ( \beta_1,\ldots,\beta_n )$ of~$\bF_{q^n}$ over~$\bF_q$:} which is equivalent to choosing a random invertible~$n\times n$ matrix over~$\bF_q$.
	\end{enumerate}
	
	\item[Public key:] Alice publishes~$\boldM(\boldnu,\boldbeta)=(\boldM^{(1)},\ldots,\boldM^{(n)})$.
	
	\item[Encryption:] The message to be encrypted is seen as an equivalence class in~$\cQ_k$. Bob chooses arbitrary~$\bolda=(a_1,\ldots,a_k)$ and~$\boldb=(b_1,\ldots,b_k)$ that correspond to his message (i.e., such that~$\bolda^\intercal \boldb$ is in the corresponding equivalence class in~$\cQ_k$), and sends $E(\bolda,\boldb)\triangleq\left(\bolda \boldM^{(i)} \boldb^\intercal\right)_{i=1}^n$ to Alice.
	
	\item[Decryption:] Alice assembles 
	\begin{align*}
	\sum_{i=1}^{n}\bolda \boldM^{(i)} \boldb^\intercal\cdot \beta_i&=\bolda \boldM(\boldnu) \boldb^\intercal=\bolda\boldnu^\intercal\boldnu\boldb^\intercal=\left(\sum_{i=1}^{k}a_i\nu_i\right)\left(\sum_{i=1}^{k}b_i\nu_i\right)\triangleq ab\;.
	\end{align*}
	Since~$a$ and~$b$ are in the Sidon space~$V$, they can be retrieved from~$ab$ up to order and up to a multiplicative factor from~$\bF_q$ (see Section~\ref{section:Preliminaries}). The respective~$\bolda$ and~$\boldb$ are then retrieved by representing~$a$ and~$b$ over~$\boldnu$. Since~$\bolda$ and~$\boldb$ correspond to a unique equivalence class in~$\cQ_k$, it follows that they determine the message sent by Bob uniquely. 
\end{description}

\ifePrint
An alternative scheme which employs randomization is given in \nameref{appendix:randomizedEncryption}. 
\fi
One clear advantage of the above system is that its \textit{information rate} approaches~$1$ as~$k$ grows. The information rate is defined as the ratio between the number of bits in Bob's message and the number of bits that are required to transmit the corresponding cyphertext. 
Due to the size of~$\cQ_k$, given earlier,
it follows that the number of information bits in Bob's message approaches~$2k\log_2q$ as~$k$ grows; this is identical to the number of information bits in the cyphertext~$E(\bolda,\boldb)$. 

On the other hand, a clear disadvantage is that the public key is relatively large in comparison with the size of the plaintext; due to the symmetry of the coefficient matrices, the public key contains $k^2(k+1)$ elements\footnote{That is,~$n=2k$ matrices, each containing~$\binom{k+1}{2}$ elements.} in~$\bF_q$, whereas the plaintext contains approximately~$2k$ elements in~$\bF_q$. This disadvantage is apparent in some other MPKCs as well. For instance, in the ABC cryptosystem~\cite[Sec.~3]{ABC}, to transmit a message of~$k$ field elements, $2k$ quadratic polynomials in~$k$ variables are evaluated. Hence, the information rate is~$\frac{1}{2}$, and in order to transmit~$k$ field elements, a public key of~$k^2(k+1)$ field elements is required. Our system suffers from a large public key as many other MPKCs, albeit at information rate which approaches~$1$.

\begin{remark}[A note about performance]\label{remark:performance}
Both encoding and decoding require only elementary operations over finite fields. Given~$\bolda$ and~$\boldb$, Bob encrypts by computing~$n$ bi-linear transforms in~$O(k^3)$. Given the cypertext, Alice obtains~$ab$ using~$O(k^2)$ operations, and follows the factorization algorithm from Section~\ref{section:Preliminaries}. This algorithm includes change-of-basis to~$\{1,\gamma\}$, which is equivalent to solving a linear equation, followed by applying a pre-determined linear transform~$T^{-1}$, solving a univariate quadratic polynomial over~$\bF_{q^n}$, and finally two computation of inverse (e.g., by the extended Euclidean algorithm) and two extractions of $(q-1)$'th root (e.g., by the~$O(k^3)$ algorithm of~\cite{Roots}). Overall, assuimg~$q$ is constant, both encoding and decoding require~$O(k^3)$ operations. Key generation can be done via a simple randomized process, and experimental results are given in Section~\ref{section:experiments}.
\end{remark}

\begin{remark}[A note about parameters]\label{remark:maxSpan}
	The fact that~$n=2k$ (or more generally, that~$k=\Theta(n)$) in the Sidon cryptosystem above seems to be essential to the security of the system. For example, using a max-span Sidon space~\cite[Sec.~IV]{SidonSpaces}, in which the set~$\{\nu_i\nu_j\}_{i,j\in[k]}$ is linearly independent over~$\bF_q$ and thus~$n\ge {k+1\choose 2}$, is detrimental to the security of the system---it is easy to verify that if~$V$ is a max-span Sidon space, then $\Span_{\bF_q}(\{\boldM^{(i)}\}_{i=1}^n)$ is the set of all~$k\times k$ symmetric matrices over~$\bF_q$. Hence, given~$E(\bolda,\boldb)=( \bolda \boldM^{(i)}\boldb^\intercal )_{i=1}^n$, by using linear operations one can have $( \bolda \boldC_{i,j}\boldb^\intercal )_{i,j\in[k]}$, where~$\boldC_{i,j}$ is a matrix which contains~$1$ in its~$(i,j)$-th entry, $1$ in its~$(j,i)$-th entry, and zero elsewhere, and as a result the expressions~$\{ a_ib_i \}_{i=1}^k$ and~$\{ a_ib_j+a_jb_i \}_{i>j}$ are obtained. Clearly, these values are the coefficients of~$p_a\cdot p_b$, where 
	\begin{align*}
	p_a(x_1,\ldots,x_k)&\triangleq \sum_{i\in[k]}a_ix_i,&\mbox{ and }
	\hspace{0.9cm}p_b(x_1,\ldots,x_k)\triangleq \sum_{i\in[k]}b_ix_i\;,
	\end{align*}	
	and thus~$\bolda$ and~$\boldb$ could be identified by factoring $p_a\cdot p_b$.
\end{remark}

\section{MinRank Attacks}\label{section:MinRankAttacks}
In what follows, we consider several attacks that are based on the well-known NP-complete problem\footnote{Or more precisely, the \textit{square} MinRank \textit{search} problem.} MinRank~\cite{MinRankCryptanalysis}. In all of these attacks, it is shown that breaking the system requires solving some special case of MinRank, and the feasibility of success is discussed.
\begin{description}
	\item[The MinRank problem.] 
	\item[Input:] Integers~$k,n,r$ and linearly independent matrices~$\boldN^{(0)},\boldN^{(1)},\ldots,\boldN^{(n)}$ in~$\bF^{k\times k}$ for some field~$\bF$.
	\item[Output:] A tuple~$(\lambda_1,\ldots,\lambda_n)\in\bF^n$, not all zero, such that\[\rank_{\bF}\left(\sum_{i=1}^{n}\lambda_i\boldN^{(i)}-\boldN^{(0)}\right)\le r.\]
\end{description}

In this section, the purpose of the attacker Eve is to find an equivalent secret key~$V'$. That is, Eve begins by guessing an irreducible polynomial~$P_E(x)$ of degree~$n=2k$ to define~$F_E=\bF_q[x]\bmod(P_E(x))=\bF_{q^n}$, where~$(P_E(x))$ is the ideal generated by~$P_E(x)$ in~$\bF_q[x]$. Then, since there exists a field isomorphism~$f:F_A\to F_E$, and since~$\nu_s\nu_t=\sum_{i=1}^n(\boldM^{(i)})_{s,t}\beta_i$ by the definition of the system, it follows that
\begin{align}\label{equation:isomorphism}
	f(\nu_s\nu_t)=f(\nu_s)f(\nu_t)=f\left(\sum_{i=1}^{n}\boldM^{(i)}_{s,t}\beta_i\right)=\sum_{i=1}^{n}\boldM^{(i)}_{s,t}f(\beta_i).
\end{align}
Namely, the tuple~$(f(\beta_i))_{i=1}^n$ is a solution to the MinRank problem whose parameters are~$r=1$, $\boldN^{(0)}=0$, $n=2k$, $\bF=F_E$, and~$\boldN^{(i)}=\boldM^{(i)}$ for~$i\in[n]$. Then, factoring the resulting rank one matrix~$\sum_{i=1}^{n}\boldM^{(i)}f(\beta_i)$ to~$f(\boldnu)^\intercal f(\boldnu)$ enables Eve to find~$V'=f(V)$, i.e., the subspace in~$F_E$ which is isomorphic to~$V$ in~$F_A$. With~$f(V)$ at her disposal, Eve may break the cryptosystem. 

To the best of the authors' knowledge, this solution is not necessarily unique; furthermore, it is unclear if breaking the system is possible if a solution is found which is not a basis of~$F_E$ over~$\bF_q$, or if the resulting~$V'$ is not a Sidon space. Nevertheless, we focus on the hardness of finding \textit{any} solution. Moreover, due to~\eqref{equation:isomorphism}, for convenience of notation we omit the isomorphism~$f$ from the discussion.

\subsection{The Kernel Attack}\label{section:kernelAttack}
The kernel formulation of MinRank relies on the fact that any nonzero vector~$\boldv\in\bF_{q^n}^k$ in $K\triangleq\ker_{\bF_{q^n}}(\sum_{i\in[n]}\beta_i\boldM^{(i)})$ gives rise to~$k$ $\bF_{q^n}$-linear equations in~$y_1,\ldots,y_n$, namely, the~$k$ equations given by~$(\sum_{i=1}^{n}y_i\boldM^{(i)} )\boldv^\intercal=0$. To find the correct $y_1,\ldots,y_n\in\bF_{q^n}$, sufficiently many~$\boldv$'s in~$K$ must be found. For example, finding~$\boldv_1\in K$ yields~$k$ equations in~$n=2k$ variables, and hence there are at least~$(q^n)^k$ possible solutions, depending on the linear dependence of these equations. Finding an additional~$\boldv_2\in K$ adds another~$k$ equations, which are likely to reduce the number of possible values for~$y_1,\ldots,y_n$ further. By repeating this process, the attacker wishes to reduce the dimension of the solution space sufficiently so that the solution~$y_1,\ldots,y_n$ could be found.

Since~$K$ is unknown, the attacker resorts to uniformly random guesses of $\boldv_1,\boldv_2\in\bF_{q^n}^k$, hoping to get them both in~$K$. However, since~$\dim K=k-1$, it follows that
\begin{align}\label{equation:kernelAttackProb}
	\Pr_{\boldv\in\bF_{q^n}^k}(\boldv\in K)=\frac{|K|}{|\bF_{q^n}^k|}=\frac{(q^n)^{k-1}}{(q^n)^k}=\frac{1}{q^n},
\end{align}
and hence the probability of finding even \textit{a single}~$\boldv\in K$ is exponentially small in the message length. 

\begin{remark}[Kernel attack over the base field]\label{remark:kernelOverBasefield}
Recall that $\boldM(\boldnu)=\sum_{i\in[n]}\beta_i\boldM^{(i)}$. In order to make the above attack feasible, one may suggest to guess nonzero vectors~$\boldv\in\bF_q^k$ rather than~$\boldv\in\bF_{q^n}^k$. However, it is easy to see that for any nonzero vector~$\boldv\in\bF_q^k$ we have $\boldM(\boldnu) \boldv\neq 0$, and in fact it is a vector with no nonzero entries. Indeed, $\boldM(\boldnu)$ is the multiplication table of $\boldnu = (\nu_1,\ldots,\nu_k)$, which is a basis of the Sidon space~$V$. Hence, $\boldM(\boldnu)\boldv$ is a vector whose $i$'th coordinate equals $\nu_i (\sum_{j\in[k]} v_j\nu_j)$. Since  the $\nu_j$'s are linearly independent over $\mathbb{F}_q$ and~$\boldv$ is nonzero, the second term in the product is nonzero, and hence so is the product.

\end{remark}

\begin{remark}[Kipnis-Shamir formulation] In a variant of this attack, one guesses kernel vectors in a systematic form, rather than in a general form. That is, the system
\begin{align*}
    \left( \sum_{i=1}^n y_i\boldM^{(i)} \right)
    \begin{pmatrix}
    1 & 0 & \ldots & 0\\
    0 & 1 & \ldots & 0\\
    \vdots & \vdots &\ddots &\vdots\\
    0 & 0 & \ldots & 1\\
    z_1 & z_2 & \ldots & z_{k-1}
    \end{pmatrix}=0.
\end{align*}
has a solution with~$y_1,\ldots,y_n,z_1,\ldots,z_{k-1}\in\bF_{q^n}$ (technically, the position of the non-unit row in the r.h.s matrix can be arbitrary; however, this can be amended by repeating the algorithm~$k$ times with different positions, or by random guessing). Similar to the attack above, one can guess two column vectors from the r.h.s matrix, and solve the resulting system, which is linear in the~$y_i$'s. However, it is readily verified that the probability to guess each~$z_i$ correctly is~$q^{-n}$, and hence analysis similar to~\eqref{equation:kernelAttackProb} applies. Alternatively, one can treat both the~$y_i$'s and the~$z_i$'s as variables over~$\bF_{q^n}$, and solve the resulting quadratic system using Gr\"{o}bner basis algorithms. Very recently~\cite{KSformulation2,KSformulation1}, it was shown that in some parameter regimes, such Gr\"{o}bner basis algorithms admit an inherent structure that can be utilized to reduce the computation time, often significantly (e.g., for the HFE cryptosystem). Whether the Sidon cryptosystem admits a similar structure remains to be studied.
\end{remark}

\subsection{The Minor Attack}\label{section:minorAttack}
In the minor attack of MinRank, one considers the system of homogeneous quadratic equations given by setting all~$2\times 2$ minors of~$\sum_{i\in[n]}y_i\boldM^{(i)}$ to zero, and (usually) solves by linearization. That is, the system is considered as a linear one in the~$\binom{n+1}{2}$ variables~$\{ z_{i,j} \}_{i\le j, i,j\in[n]}$, where~$z_{i,j}$ represents~$y_iy_j$ for every~$i$ and~$j$. The resulting homogeneous linear system has a right kernel of dimension at least one; if it happens to be at most one, the attacker finds a nonzero solution~$\boldw=(w_{i,j})_{i\le j}$ and arranges it in a symmetric matrix
\begin{align}\label{equation:vectorToMatrix}
	\matricize(\boldw)=\begin{pmatrix}
		w_{1,1} & w_{1,2} & \ldots & w_{1,n}\\
		w_{1,2} & w_{2,2} & \ldots & w_{2,n}\\
		\vdots & \cdots & \ddots & \vdots \\
		w_{1,n} & w_{2,n} & \cdots & w_{n,n}
	\end{pmatrix}.
\end{align} 
Then, the attacker finds a rank one decomposition~$(w_1,\ldots,w_n)^\intercal(w_1,\ldots,w_n)$ of~\eqref{equation:vectorToMatrix} (which is guaranteed to exist, since the solution~$z_{i,j}=y_iy_j$ has a rank one decomposition, and the dimension of the kernel is one), which provides a solution. 

In most systems the dimension of the kernel will indeed be at most one. Otherwise the attacker is left with yet another MinRank problem, that we call \textit{secondary}, in which a ``rank-one vector'' (that is, a vector~$\boldw$ such that~$\matricize(\boldw)$ is of rank one) must be found in the kernel. In what follows it is shown that this attack on the Sidon cryptosystem results in the latter scenario. That is, attempting to solve the minor attack via linearization results in a linear system with a large kernel. Moreover, it is shown that the secondary (and tertiary, etc.) attack suffers from the same effect.

Let~$\boldOmega$ be the quadratic system which results from setting all~$2\times 2$ minors of~$\sum_{i\in[n]}y_i\boldM^{(i)}$ to zero. This system contains~$\binom{k}{2}^2$ equations, each is a linear combination over~$\bF_q$ of the~$\binom{n+1}{2}$ monomials $y_1^2,\ldots,y_n^2,y_1y_2,\ldots,y_{n-1}y_{n}$. To break the cryptosystem, the values of~$y_1,\ldots,y_{n}$ in a solution to~$\boldOmega$ should form a basis to~$\bF_{q^n}$ over~$\bF_q$, and as discussed earlier, it is unclear if the system can be broken otherwise. Yet, for generality we focus on the hardness of finding \textit{any} solution. 

Let~$\boldOmega_{\text{lin}}$ be the matrix which results from linearizing~$\boldOmega$. That is, each of the~$\binom{n+1}{2}$ columns of~$\boldOmega_{\text{lin}}$ is indexed by a monomial~$y_sy_t$, and each row is indexed by a minor~$((i,j),(\ell,d))$ (i.e., the minor that is computed from the~$i$'th and~$j$'th rows and the~$\ell$'s and~$d$'th columns). The value of an entry in column~$y_sy_t$ and row~$((i,j),(\ell,d))$ is the coefficient of~$y_sy_t$ in the equation for the~$2\times 2$ minor of $(\sum_{i\in[n]}y_i\boldM^{(i)})$ in rows~$i$ and~$j$, and columns~$\ell$ and~$d$. Note that a solution to~$\boldOmega$ corresponds to a vector in the right kernel of~$\boldOmega_{\text{lin}}$, but the inverse is not necessarily true. 

We begin by discussing several aspects of~$\boldOmega_{\text{lin}}$. First, since the matrices~$\boldM^{(i)}$ are symmetric, many rows in~$\boldOmega_{\text{lin}}$ are identical; minor~$((i,j),(\ell,d))$ identical to minor~$((\ell,d),(i,j))$. Hence, the effective number of rows is at most
\begin{align}\label{equation:symmetryReduction}
	\binom{\binom{k}{2}+1}{2}.
\end{align}
Second, $\boldOmega_{\text{lin}}$ is over~$\bF_q$, while the required solution is in~$\bF_{q^n}$. One way to circumvent this is by representing every~$y_i$ using~$n$ variables over~$\bF_q$. The resulting linearized system can be described using Kronecker products. By using the fact that the rank of a Kronecker product is the product of the individual ranks, it can be easily shown that this system has a large kernel, and thus solving by linearization is not feasible. The full details of this approach are given in  
\ifePrint
\nameref{appendix:LinearizationAttack}.
\fi

More importantly, in contrast to Remark~\ref{remark:kernelOverBasefield}, one can simply find~$\ker_{\bF_q}(\boldOmega_{\text{lin}})$;
since~$\rank_{\bF_q}(\boldOmega_{\text{lin}})=\rank_{\bF_{q^n}}(\boldOmega_{\text{lin}})$, it follows that the true solution~$(z_{i,j})_{i\le j}=(\beta_i\beta_j)_{i\le j}$ lies in~$\Span_{\bF_{q^n}}(\ker_{\bF_q}(\boldOmega_{\text{lin}}))$. Put differently, one can solve~$\boldOmega$ via linearization over~$\bF_q$ (i.e., obtain an~$\bF_q$-basis to~$\boldOmega_{\text{lin}}$'s right kernel), span it over~$\bF_{q^n}$, and search for a rank one vector. However, in what follows it is shown that the rank of~$\boldOmega_{\text{lin}}$ is low (specifically, $\rank(\boldOmega_{\text{lin}})\le \binom{n+1}{2}-n$), and hence this approach is not feasible either. One might wonder if the secondary MinRank problem that emerges is itself solvable by linearization, for which we show that the answer is negative, and the proof is similar.

\subsubsection{Bounding the rank of~$\boldOmega_{\text{lin}}$.}
Let~$\boldnu=(\nu_1,\ldots,\nu_k)$ be the secret basis of~$V$ and let $\boldu=(\nu_1,\ldots,\nu_n)$ be an extension of~$\boldnu$ to a complete   basis of~$\bF_{q^n}$ over~$\bF_q$.  Let~$\boldbeta=(\beta_1,\ldots,\beta_n)$ be the secret basis of~$\bF_{q^n}$. Therefore, we have that~$\boldu^\intercal\boldu=\sum_{i\in[n]}\beta_i\boldN^{(i)}$ for some matrices~$\boldN^{(i)}\in\bF_q^{n\times n}$. It is readily verified that for every~$i\in[n]$, the upper left~$k\times k$ submatrix of~$\boldN^{(i)}$ is the public key coefficient matrix~$\boldM^{(i)}$. 
In addition, let~$\boldE\in\bF_q^{n\times n}$ be the change-of-basis matrix such that~$\boldbeta=\boldu\boldE$, and then
\begin{equation}\label{stam}
    \boldbeta^\intercal\boldbeta=\boldE^\intercal\boldu^\intercal\boldu\boldE=\sum_{i\in[n]}\beta_i\boldE^\intercal\boldN^{(i)}\boldE.
\end{equation}

Construct a system~$\boldGamma$ of quadratic equations in the variables~$y_1,\ldots,y_n$ by setting all the~$2\times 2$ minors of~$\sum_{i\in[n]}y_i\boldN^{(i)}$ to zero, and let $\boldGamma_{\text{lin}}$ be the linearization of the set of equations in $\boldGamma$, obtained by replacing each monomial $y_iy_j$ by the variable $z_{i,j}$. Notice that one obtains~$\boldOmega_{\text{lin}}$ from~$\boldGamma_{\text{lin}}$ by omitting every row~$((i,j),(\ell,d))$ of~$\boldGamma_{\text{lin}}$ with either one of~$i,j,\ell,d$ larger than~$k$. Therefore, it follows that~$\ker (\boldGamma_{\text{lin}})\subseteq \ker(\boldOmega_{\text{lin}})$.

We claim that each matrix $\boldE^\intercal\boldN^{(l)}\boldE, l\in [n]$ defines a valid solution to $\boldGamma_{\text{lin}}$ simply by setting $z_{i,j}=(\boldE^\intercal\boldN^{(l)}\boldE)_{i,j}=(\boldE^\intercal\boldN^{(l)}\boldE)_{j,i}$. Then, it will be shown that~$\{\boldE^\intercal\boldN^{(l)}\boldE\}_{l\in [n]}$ are linearly independent, and thus so are the solutions they define. This would imply that the dimension of the solution space of~$\boldGamma_{\text{lin}}$ is at least~$n$, and since $\ker (\boldGamma_{\text{lin}})\subseteq \ker(\boldOmega_{\text{lin}})$, it would also imply that the dimension of the solution space of~$\boldOmega_{\text{lin}}$ is at least~$n$.

For an element $\alpha\in \mathbb{F}_{q^n}$ and $i\in [n]$ let $(\alpha)_i\in \mathbb{F}_q$ be the  coefficient of $\beta_i$ in the expansion of $\alpha$ as a linear combination of the $\beta_j$'s over $\mathbb{F}_q$, i.e.,~$\alpha=\sum_{i\in[n]}(\alpha)_i\beta_i$. Then, it follows from the definition of the~$\boldN^{(l)}$'s that $\boldN^{(l)}_{i,j}=(\nu_i\nu_j)_l$. Similarly, it follows from~\eqref{stam} that~$(\boldE^\intercal\boldN^{(l)}\boldE)_{i,j}=(\beta_i\beta_j)_l$.
\begin{lemma}\label{lemma:matricesinKerGamma}
For every $l\in [n]$ the assignment $z_{i,j}=(\boldE^\intercal\boldN^{(l)}\boldE)_{i,j}$ is a solution for    $\boldGamma_{\text{lin}}$. 
\end{lemma}
\begin{proof}
    Let 
    $\begin{pmatrix}
a & b \\
c & d 
\end{pmatrix}\in\bF_{q^n}^{2\times 2}$ be an arbitrary~$2\times 2$ submatrix of $ \boldu^\intercal\boldu=\sum_{i\in[n]}\beta_i\boldN^{(i)}$. First, notice that the respective equation in~$\boldGamma$ is
\begin{align*}
	\left( \sum_{i\in[n]}(a)_iy_i \right)\left( \sum_{i\in[n]}(d)_iy_i \right)-\left( \sum_{i\in[n]}(b)_iy_i \right)\left( \sum_{i\in[n]}(c)_iy_i \right)=0,
\end{align*}
which after linearization becomes
\begin{align}\label{equation:linearizredinGamma}
	\sum_{i,j\in[n]}(a)_i(d)_jz_{i,j}-\sum_{i,j\in[n]}(b)_i(c)_jz_{i,j}=0.
\end{align}

Second, since $\boldu^\intercal\boldu$ is a rank one matrix, so is any of its~$2\times 2$ submatrices, and therefore $ad-bc=0$. Since this implies that $(ad-bc)_l=0$ for every~$l\in[n]$, it follows that
\begin{align*}
0&=(ad-bc)_l=(ad)_l-(bc)_l\\
&=\left(\textstyle{\sum}_{i,j\in[n]}(a)_i(d)_j\beta_i\beta_j\right)_l-\left(\textstyle{\sum}_{i,j\in[n]}(b)_i(c)_j\beta_i\beta_j\right)_l\\
&=\textstyle{\sum}_{i,j\in[n]}(a)_i(d)_j(\beta_i\beta_j)_l-\textstyle{\sum}_{i,j\in[n]}(b)_i(c)_j(\beta_i\beta_j)_l\\
&=\textstyle{\sum}_{i,j\in[n]}(a)_i(d)_j(\boldE^\intercal\boldN^{(l)}\boldE)_{i,j}-\textstyle{\sum}_{i,j\in[n]}(b)_i(c)_j(\boldE^\intercal\boldN^{(l)}\boldE)_{i,j}.\\
\end{align*}
Therefore, it follows from~\eqref{equation:linearizredinGamma} that for every~$l\in[n]$, the assignments  $z_{i,j}=(\boldE^\intercal\boldN^{(l)}\boldE)_{i,j}$ is a zero of~$\boldGamma_{\text{lin}}$, as needed.
\end{proof}
\begin{lemma}
    The $n$ matrices $\boldE^\intercal\boldN^{(l)}\boldE,l\in [n]$ are linearly independent over~$\bF_q$.
\end{lemma}
\begin{proof}
Since~$\boldE$ is invertible, the claim follows by showing that the matrices $\{\boldN^{(l)}\}_{l\in [n]}$ are linearly independent over~$\bF_q$. For $l\in [n]$ let~$\bolda=(a_i)_{i\in[n]}$ and~$\boldb=(b_i)_{i\in[n]}$ be nonzero vectors in~$\bF_q^n$ such that $(\sum_{i\in[n]} a_i\nu_i)(\sum_{j\in[n]} b_j\nu_j)=\beta_l$. Then, it follows that
\begin{align*}
 \beta_l= \left(\textstyle{\sum}_{i\in[n]} a_i\nu\right)\left(\textstyle{\sum}_{j\in[n]} b_j\nu_j\right)
 =\bolda\boldu^\intercal\boldu \boldb^\intercal=\sum_{i\in[n]}\beta_i\bolda\boldN^{(i)}\boldb^\intercal,
\end{align*}
\end{proof}
and hence
$$
\bolda\boldN^{(i)}\boldb^\intercal=
\begin{cases}
1   & i=l\\
0 & i\neq l
\end{cases}.$$
This readily implies that every~$\boldN^{(l)}$ is linearly independent of the remaining matrices~$\{\boldN^{(j)}\}_{j\ne l}$; otherwise, a nontrivial linear combination~$\boldN^{(l)}=\sum_{j\ne l}\alpha_i\boldN^{(j)}$ would imply that~$1=0$ by multiplying from the left by~$\bolda$ and from the right by~$\boldb^\intercal$. Since this holds for every~$l\in[n]$, the claim follows.

\begin{remark}\label{remark:extensionsDontWork}
	We have experimentally verified for a wide range of~$q$ and~$n$ values that~$\rank(\boldOmega_{\text{lin}})=\binom{n+1}{2}-2n$, namely, that~$\dim\ker(\boldOmega_{\text{lin}})=2n$. In the above it is proved that~$\dim\ker(\boldOmega_{\text{lin}})\ge n$, and the remaining~$n$ dimensions remain unexplained. \ifePrint One might conjecture that different extensions of~$\nu_1,\ldots,\nu_k$ to~$\boldu$ might result in different kernels of~$\boldGamma_{\text{lin}}$, which might explain the missing~$n$ dimensions in~$\ker(\boldOmega_{\text{lin}})$. However, it is shown in \nameref{appendix:difBasisExt} that this is not the case, and all possible extensions of~$\nu_1,\ldots,\nu_k$ to~$\boldnu$ result in identical~$\ker(\boldGamma_{\text{lin}})$.\fi
\end{remark}

\subsubsection{Secondary minor attack.} 
In the above it is shown that by attempting to solve the minor attack via linearization, one is left with yet another MinRank problem, which we call \textit{secondary}. That is, in the secondary problem one must find a rank one vector in the~$\bF_{q^n}$-span of~$\ker(\boldOmega_{\text{lin}})$ (i.e., a rank one matrix in~$\{\matricize(\boldy)\vert \boldy\in \Span_{\bF_{q^n}}(\ker_{\bF_q}(\boldOmega_{\text{lin}})) \}$, where $\matricize(\cdot)$ is defined in~\eqref{equation:vectorToMatrix}). To show the hardness of the \textit{primary} minrank attack, it was shown earlier that it is not feasible to find a rank one matrix in the~$\bF_{q^n}$-span of~$\{\boldN^{(i)}\}_{i\in[n]}$ via linearization. According to Lemma~\ref{lemma:matricesinKerGamma}, to show that hardness of the \textit{secondary} attack it suffices to show that that it is not feasible to find a rank one matrix in the~$\bF_{q^n}$-span of~$\{ \boldE^\intercal\boldN^{(i)}\boldE \}_{i\in[n]}$. Since~$\boldE$ is invertible, it readily follows that a solution to the primary attack is also a solution to the secondary, and vice versa. Therefore, solving the secondary minor attack via linearization is not feasible either.

Moreover, in the secondary attack and in the subsequent ones (tertiary, quaternary, etc.), we observe the following intriguing circular phenomenon. Let $\{\boldB^{(i)}\}_{i\in[n]}\subseteq\bF_q^{n\times n}$ such that $\boldbeta^\intercal\boldbeta=\sum_{i\in[n]}\beta_i\boldB^{(i)}$. Since~$\boldbeta=\boldu\boldE$ and~$\boldu^\intercal\boldu=\sum_{i\in[n]}\beta_i\boldN^{(i)}$, it follows that
\begin{align*}
	\boldu^\intercal\boldu=(\boldE^{-1})^\intercal\boldbeta^\intercal\boldbeta\boldE^{-1}=\sum_{i\in[n]}\beta_i(\boldE^{-1})^\intercal\boldB^{(i)}\boldE^{-1}=\sum_{i\in[n]}\beta_i\boldN^{(i)},
\end{align*}
and hence~$\boldE^\intercal\boldN^{(i)}\boldE=\boldB^{(i)}$. 
That is, in the secondary attack one should find an $\bF_{q^n}$-assignment to~$y_1,\ldots,y_n$ so that~$\sum_{i\in[n]}y_i\boldB^{(i)}$ is of rank one. Then, one repeats the proof of hardness for~$\sum_{i\in[n]}y_i\boldN^{(i)}$ for the special case where~$\boldu=\boldbeta$, i.e., where~$\boldN^{(i)}=\boldB^{(i)}$ and~$\boldE=\boldI$. Lemma~\ref{lemma:matricesinKerGamma} then implies that~$z_{i,j}=(\boldB^{(l)})_{i,j}$ is in the kernel of the linearized system, for every~$\ell\in[n]$. Consequently, while attempting to solve the secondary attack by linearization, one encounters a \textit{tertiary} attack, where one should find an~$\bF_{q^n}$ assignment to~$y_1,\ldots,y_n$ so that~$\sum_{i\in[n]}y_i\boldB^{(i)}$ is of rank one. Clearly, this tertiary attack is \textit{identical} to the secondary one. Moreover, by following the same arguments we have that all subsequent attacks (quaternary, quinary, etc.) are identical to the secondary one.

\begin{remark}
	As mentioned earlier, we have verified experimentally for a wide range of~$q$ and~$k$ values over many randomized constructions, that~$\dim\ker(\boldOmega_{\text{lin}})=2n$, but as of yet have not been able to explain that mathematically. In the context of the secondary attack, one might suggest to take a basis~$\boldv_1,\ldots,\boldv_{2n}$ of~$\ker(\boldOmega_{\text{lin}})$, and search for a rank one matrix in the~$\bF_{q^n}$-span of~$\{ \matricize(\boldv_i) \}_{i\in[2n]}$, again using linearization. We have verified experimentally that the resulting system is of the same rank as~$\boldOmega_{\text{lin}}$, hence not feasibly solvable via linearization, albeit having~$\binom{2n+1}{2}$ columns rather than~$\binom{n+1}{2}$. 
\end{remark}

\section{Other Attacks}\label{section:OtherAttacks}
\subsection{Finding a structured Sidon space}\label{section:structured}
In this section we present an attack which is specific to the structure of the Sidon space~$V$ from Construction~\ref{construction:n=2k}. By guessing an alternative construction of~$\bF_{q^n}$, Eve may assemble a certain set of polynomial equations, which is guaranteed to have a solution. Each such solution defines a subspace~$V'$, most likely a Sidon space, whose coefficient matrices are identical to those of the secret Sidon space~$V$, and hence, it can be used to break the system. However, the resulting equation set is only slightly underdetermined, and hence it is unlikely that a suitable polynomial algorithm exists.

In this attack, Eve guesses an irreducible polynomial~$P_E(x)\in \bF_q[x]$ of degree~$n$, and constructs~$\bF_{q^n}$ as~$F_E\triangleq \bF_q[x]\bmod (P_E(x))$, where~$(P_E(x))$ denotes the ideal in~$\bF_q[x]$ which is generated by~$P_E(x)$. Further, she guesses a basis~$\omega_1,\ldots,\omega_n$ of~$F_E$ over~$\bF_q$ such that~$\omega_1,\ldots,\omega_k$ is a basis for~$G_E$, the unique subfield of size~$q^k$ of~$F_E$. 

To find~$\boldnu'\triangleq(\nu_1',\ldots,\nu_k')\in\bF_{q^n}^k$ and~$\boldbeta'\triangleq(\beta_1',\ldots,\beta_n')\in\bF_{q^n}^n$ such that $\boldM(\boldnu,\boldbeta)=\boldM(\boldnu',\boldbeta')$, Eve defines variables~$\{ u_{i,j} \}_{i,j\in[k]}$, $\{ b_{i,j} \}_{i,j\in[n]}$, and~$\{g_i\}_{i=1}^n$, all of which represent elements in~$\bF_q$, and
\begin{align*}
\gamma'&\triangleq\sum_{j=1}^{n}g_j\omega_j\;,\\
\nu_i'   &\triangleq \left( \sum_{j=1}^ku_{i,j}\omega_j \right)+\left( \sum_{j=1}^{n}g_j\omega_j \right)\left( \sum_{j=1}^ku_{i,j}\omega_j \right)^q\\
&=  \left( \sum_{j=1}^ku_{i,j}\omega_j \right)+\left( \sum_{j=1}^{n}g_j\omega_j \right)\left( \sum_{j=1}^ku_{i,j}\omega_j^q \right)\mbox{ for all }i\in[k],\mbox{ and}
\end{align*}
\begin{align*}
\beta_i'&\triangleq\sum_{j=1}^{n}b_{i,j}\omega_j\mbox{ for all }i\in[n]\;.
\end{align*}
Eve then defines the following~${k+1\choose 2}$ equations over~$\bF_{q^n}$,
\begin{align}\label{equation:degree4}
\nu_s'\nu_t' = \sum_{i=1}^{n}M^{(i)}_{s,t}\beta_i'\mbox{ for all }s,t\in[k],~s\ge t\;.
\end{align}
Finally, by expressing each side of every equation as a linear combination of\linebreak$\{ \omega_i \}_{i\in[n]}$ and comparing coefficients, Eve obtains~$n\cdot {k+1\choose 2}=k^2(k+1)$ equations over~$\bF_q$ in~$n^2+k^2+n=5k^2+2k$ variables. The left hand sides of these equations are polynomials in~$k^2+n=k^2+2k$ variables and degree four, and the right hand sides are linear polynomials in $n^2=4k^2$ variables. 

Since the isomorphism~$f$ exists~\eqref{equation:isomorphism}, the system is guaranteed to have a solution. The resulting subspace~$V'\triangleq \Span\{ \nu_i' \}_{i\in[n]}$ is a Sidon space if the corresponding~$\gamma'$ satisfies the conditions of Construction~\ref{construction:n=2k}. However, it seems that the straightforward algorithms for obtaining a solution are infeasible.

Notice that the terms on the left hand side of~\eqref{equation:degree4} are of either of the forms
\begin{align*}
u_{s,t}u_{\ell,r},~g_ju_{s,t}u_{\ell,r},\mbox{ or }g_ig_ju_{s,t}u_{\ell,r},
\end{align*}
for $s,t,\ell,r\in[k]$ and~$i,j\in[n]$. Hence, a straightforward reduction to the quadratic case includes replacing those terms by $u_{s,t,\ell,r}$, $g_j\cdot u_{s,t,\ell,r}$, and $g_{i,j}u_{s,t,\ell,r}$, respectively. In the resulting quadratic equation set, the number of equations remains~$e\triangleq k^2(k+1)$. The number of variables however, comprises of~$k^4$ variables of the form~$u_{s,t,\ell,r}$, $4k^2$ of the form~$g_{i,j}$, $4k^2$ of the form~$b_{i,j}$, and $2k$ of the form~$g_j$. Thus, the overall number of variables is~$v\triangleq k^4+8k^2+2k$ and the equation set is~\textit{underdetermined} ($e<v$), with~$v=\Theta(e^{4/3})$. 

Algorithms for solving underdetermined systems of quadratic equations were studied in~\cite{underdetermined3,underdetermined2,underdetermined1}. It is known that highly underdetermined systems ($v=\Omega(e^2)$) and highly overdetermined systems ($e=\Omega(v^2)$) are solvable in randomized polynomial time. On the other hand, if~$e=v\pm O(1)$ then the current state-of-the-art algorithms are exponential. The results in our case ($v=\Theta(e^{4/3})$) seem inconclusive. In our experimental section it is shown that standard Gr\"{o}bner basis algorithms are far from feasible for solving this system for~$k<10$.

\subsection{Extracting the message from the cyphertext}\label{section:bilinear}
It is readily verified that extracting~$\bolda$ and~$\boldb$ from~$E(\bolda,\boldb)=(\bolda \boldM^{(i)} \boldb^\intercal)_{i=1}^n$ and $\boldM(\boldnu,\boldbeta)=(\boldM^{(i)})_{i=1}^n$ is equivalent to solving the corresponding non-homogeneous bilinear system of~$2k$ equations and~$2k$ variables. It seems that the state-of-the-art algorithm for solving a bilinear system is given by~\cite[Cor.~3]{Bilinear11}, whose complexity is
\begin{align*}
O\left({n_a+n_b+\min(n_a+1,n_b+1) \choose \min(n_a+1,n_b+1)}^\omega\right),
\end{align*}
where~$n_a$ and~$n_b$ are the number of entries in~$\bolda$ and~$\boldb$, and~$\omega$ is the exponent of matrix multiplication. However, this specialized algorithm requires homogeneity, and in any case applying it to our problem requires~$O({3k+1\choose k+1}^\omega)$, which is infeasible even for small values of~$k$. 

We also note that it is possible to apply algorithms that do not exploit the \textit{bilinear} nature of the system, but rather only its quadratic one. However, evidence show that standard Gr\"{o}bner basis algorithms for solving quadratic equations perform very poorly on quadratic equation sets in which the number of equations and the number of variables is equal.
Following Remark~\ref{remark:maxSpan}, it should be noted that if one would employ a max-span Sidon space as the private key, the resulting bilinear system has $\Theta(k^2)$ equations and~$2k$ variables, and hence it is easy to solve by~\cite[Sec.~6.5]{Courtois} and references therein.

\section{Experiments}\label{section:experiments}

Experiments were run using a computer with an Intel i7-9750H CPU with 16GB of RAM. Computations were done on an engineering cluster node with 2 Intel x86 E5 2650 processors with 64 gigabytes of RAM. For reproducibility, the code for all experiments is given~\cite{Git}. Throughout this section we denote the number of equations by~$e$ and number of variables by~$v$.

Before discussing attacks, we discuss the performance of the system itself. Encoding and decoding use simple finite field operations, and had marginal affect on run-times (see Remark~\ref{remark:performance}). The significant part of the key generation algorithm is the choice of~$\gamma$, which defines the secret Sidon space; this amounts to choosing the quadratic polynomial~$P_{a,b}$ so that it is irreducible over~$\bF_{q^k}$ with~$c\in\overline{\cW}_{q-1}$. This was done at random, and mean success times for different~$k$ and~$q$ values over~$10$ trials are given in Fig.~\ref{figure:keyGeneration}.

The easiest attack to implement seems to be the bilinear one (Sec.~\ref{section:bilinear}), due to the small size of the associated inhomogeneous bilinear system ($v=e=2k$). Specialized algorithms for improved performance on bilinear systems~\cite{Bilinear11} are inapplicable, since they require homogeneity and have exponential complexity. We used the \texttt{F4} algorithm from the \texttt{FGb} library~\cite{FGb} in combination with \texttt{FGb\_sage}~\cite{FGbSage}. The system was homogenized before the Gr\"{o}bner basis was computed. Attacks were efficiently carried out for $k \le 10$ (i.e., $v=21$ and $e=20$). The field size~$q$ was varied between~$q=3$ and~$q=65521$, but had marginal effect on running times. Past $k=10$, the \texttt{F4} algorithm exceeded the~$50\cdot 10^6$ bound on the dimension of the matrix. Average running times, that are given below in Fig.~\ref{figure:bilinearRuntimes}, are consistent with the exponential growth one would expect.

\begin{figure}
\centering
\begin{tikzpicture}
\begin{axis}[
    xlabel={$k$},
    ylabel={time (s)},
    xmin=5, xmax=40,
    ymin=0, ymax=85,
    ytick={10,20,30,40,50,60,70,80},
    xtick={10,20,30,40},
    legend pos=north west,
    ymajorgrids=true,
    grid style=dashed,
]
    \addplot[
    color=blue,
    mark=otimes,
    ]
    coordinates {
    (5,0.05752)
    (10,0.3397)
    (15,0.675)
    (20,1.992)
    (25,3.674)
    (30,7.364)
    (35,14.337)
    (40,16.4358)
    };

   \addplot[
    color=green,
    mark=diamond,
    ]
    coordinates {
    (5,0.09276)
    (10,0.4386)
    (15,1.0123)
    (20,4.605)
    (25,6.809)
    (30,16.834)
    (35,32.524)
    (40,40.818)
    };

   \addplot[
    color=red,
    mark=triangle,
    ]
    coordinates {
    (5,0.132)
    (10,0.611)
    (15,1.642)
    (20,6.069)
    (25,30.778)
    (30,55.1187)
    (35,72.74)
    (40,83.37)
    };
   \legend{$q=5$, $q=53$, $q=541$}
\end{axis}
\end{tikzpicture}\caption{Average running times of randomized key generation for various~$q$ values.}\label{figure:keyGeneration}
\end{figure}
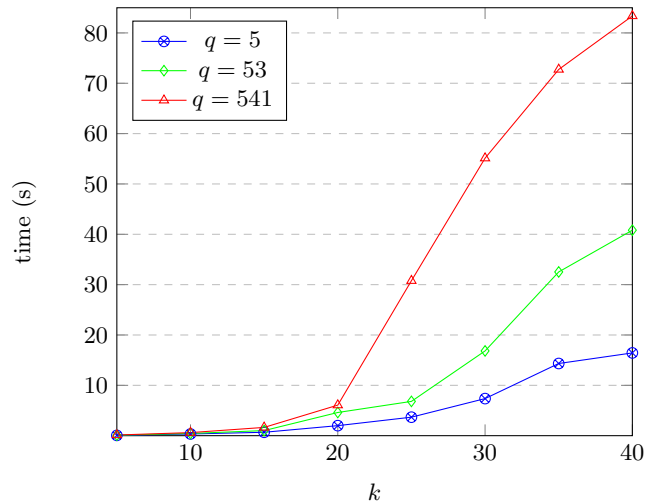

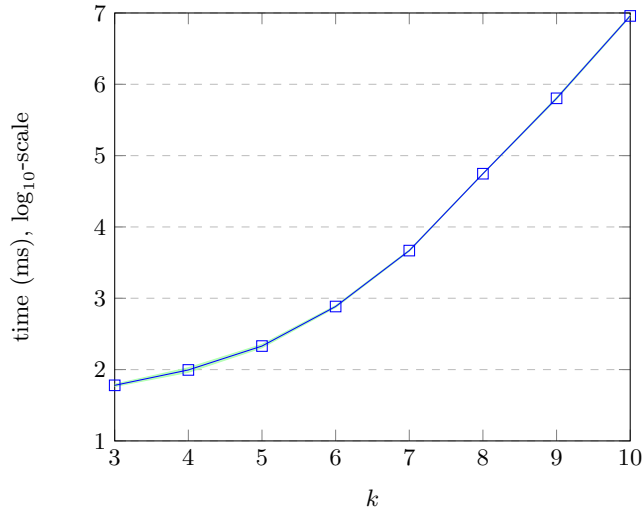
\begin{figure}
\centering
\begin{tikzpicture}
\begin{axis}[
    xlabel={$k$},
    ylabel={time (ms), $\log_{10}$-scale},
    xmin=3, xmax=10,
    ymin=1, ymax=7,
    ytick={1,2,3,4,5,6,7},
    xtick={3,4,5,6,7,8,9,10},
    ymajorgrids=true,
    grid style=dashed,
]
    \addplot[name path=us_top,color=green!30] coordinates {
    (3,1.795151659)
    (4,2.025761036)
    (5,2.353172045)
    (6,2.897338079)
    (7,3.676968306)
    (8,4.750066167)
    (9,5.816770743)
    (10,6.9685825)
    };
    \addplot[name path=us_down,color=green!30] coordinates {
    (3,1.762963985)
    (4,1.961684501)
    (5,2.30553786)
    (6,2.870945295)
    (7,3.6596305)
    (8,4.742546816)
    (9,5.789722406)
    (10,6.948308682)
    };
    \addplot[green!50,fill opacity=0.5] fill between[of=us_top and us_down];
    \addplot[
    color=blue,
    mark=square,
    ]
    coordinates {
    (3,1.779355952)
    (4,1.994903443)
    (5,2.330007701)
    (6,2.884342148)
    (7,3.668385917)
    (8,4.746322765)
    (9,5.803457116)
    (10,6.958563883)
    };
\end{axis}
\end{tikzpicture}\caption{Average running times for the bilinear attack (Sec.~\ref{section:bilinear}) on randomly chosen Sidon cryptosystems. Standard deviations are given in light shaded area, which is barely visible.}\label{figure:bilinearRuntimes}
\end{figure}

Next, executing the minor attack (Sec.~\ref{section:minorAttack}) past $k = 2$ proved difficult. The first class of algorithms considered for this attack were of the eXtended Linearization (XL) family~\cite{Courtois}, but were not promising for a number of reasons. First, XL algorithms require a unique solution, which is not the case in our systems. Second, complexity analysis shows poor asymptotic performance; XL algorithms are polynomial if~$\varepsilon\triangleq \frac{e}{v^2}$ is fixed, but are in general exponential otherwise. In our case~$\varepsilon$ approaches~$0$ as~$k$ increases, and thus we resorted to Gr\"{o}bner basis algorithms.

Experimental evidence shows that while the attack generates~$\Theta(k^5)$ equations\ifePrint~\eqref{equation:SizeOfBoldCBoldOmega},\fi~only $2k^2(2k - 3)$ of them are independent (an upper bound of~$2k^2(2k+1)$ \ifePrint is given in~\eqref{equation:rankKronecker}). \fi
Benchmarks for fast Gr\"{o}bner basis algorithms~\cite{M4GB} show that the $k = 3$ system exists on the borderline of what has been computed, and that is supported by experimental evidence. Both implementations of \texttt{F5} as well as the \texttt{FGb} library were used to try and compute Gr\"{o}bner bases for this system, but neither performed well, with the \texttt{FGb} library quickly surpassing the~$50\cdot 10^6$ matrix bound and the \texttt{F5} algorithm not terminating after several days of running. For~$k=4$ and~$k=5$ we were able to compute the degree of regularity, which was~$8$.

The structured attack in Sec.~\ref{section:structured}, which has~$v=\Theta(k^2)$ and~$e=\Theta(k^3)$ proved to be the least feasible, where~$k = 3$ ($v=36$ and $e=54$) and up were completely unsolvable. The system can be reduced to a quadratic one with~$v=\Theta(k^4)$ variables, which did not seem to accelerate the computation.

\section{Discussion}\label{section:Discussion}
\ifePrint
In this paper the Sidon cryptosystem was introduced, and several straightforward attacks were given. These attacks were shown to induce instances of several problems for which it is unlikely that a polynomial algorithm exists. Nevertheless, a finer analysis of the algebraic nature of Sidon spaces might shed some light on the structure of these instances, and consequently, might prove the system insecure. On the other hand, a rigorous proof for the hardness of the Sidon cryptosystem, which has yet to be found, will be a significant achievement in post-quantum cryptography.
\fi 

The first order of business in extending this work is finding the remaining~$n$ dimensions in the kernel of~$\boldOmega_{\text{lin}}$, and we have verified experimentally that these additional~$n$ dimensions do not exist in~$\boldGamma_{\text{lin}}$. More broadly, we suggest that applications of Sidon spaces to cryptography extend beyond what is discussed in the paper. Other than using different constructions of Sidon spaces (e.g.,~\cite{OtherSidon}) in the framework described above, we suggest to study the following concepts in order to strengthen the resulting systems.
\begin{description}
	\item[$r$-Sidon spaces.] The Sidon spaces in this paper enable the product of every two elements to be factored uniquely (up to constants). This is a special case of~$r$-Sidon spaces, in which the product of any~$r$ elements can be factored uniquely (see~\cite[Sec.~VI]{SidonSpaces}). This would extend the hardness of the system from solving bilinear systems to~$r$-linear ones.
	\item[High rank multiplication table.] Most of the susceptibility of the Sidon cryptosystem lies in the matrix~$\boldM(\boldnu)$ (see Sec.~\ref{section:TheSidonCryptosystem}) being of rank one. To remedy that, let~$U,V\in\grsmn{q}{4k}{k}$ be min-span Sidon spaces such that~$V$ is spanned by~$\boldnu=(\nu_i)_{i=1}^k$, $U$ is spanned by~$\boldupsilon=(\upsilon_i)_{i=1}^k$ and~$U^2\cap V^2=\{0\}$. It is readily verified that Bob in able to decrypt the ciphertext even if the matrix~$\boldM(\boldnu)=\boldnu^\intercal\boldnu$ is replaced by~$\boldnu^\intercal\boldnu+\boldupsilon^\intercal\boldupsilon$: Bob will begin by extracting~$\bolda\boldnu^\intercal\boldnu\boldb$ from the ciphertext, which is possible since~$U^2\cap V^2=\{0\}$, and continue similarly. If the vectors~$\boldnu$ and~$\boldupsilon$ are independent over~$\bF_{q^n}$, the resulting matrix is of rank two, and hence the system's resilience against MinRank attacks is increased.
\end{description}

 {\bf Funding.} This work was partially supported by the European Research Council (ERC grant number 852953) and by the Israel Science Foundation (ISF grant number 1030/15). 
\ifePrint
\section*{Appendix~A}\label{appendix:omittedProof}
\textbf{An omitted proof.}
\begin{lemma}\label{lemma:Qk}
	For any prime power~$q$ and an integer~$k$, the size of~$\cQ_k$ (see Section~\ref{section:TheSidonCryptosystem}) is~$\frac{(q^k-1)(q^k-q)}{2(q-1)}+q^k-1$.
\end{lemma}
\begin{proof}
	The following statements, which are easy to prove, are left as an exercise to the reader.
	\begin{enumerate}
		\item \label{item:dependence} For every~$\bolda$ and~$\boldb$ in~$\bF_q^k\setminus\{0\}$, the matrix~$\bolda^\intercal \boldb$ is symmetric if and only if~$\bolda\in\Span_{\bF_q}(\boldb)$.
		\item \label{item:sym} Every~$\bolda,\boldb$ in~$\bF_q^k\setminus\{ 0 \}$ and every~$\lambda,\mu\in\bF_q^*$ satisfy that
		$\bolda^\intercal \cdot (\lambda \bolda)=\boldb^\intercal \cdot(\mu \boldb)$ if and only if $\mu\lambda^{-1}$ is a quadratic residue, and~$\bolda=\sqrt{\mu\lambda^{-1}}\cdot \boldb$.
		\item \label{item:nonsym} Every~$\bolda,\boldb,\boldc,$ and~$\boldd$ in~$\bF_q^k$ such that~$\bolda\notin\Span_{\bF_q}(\boldb)$ and $\boldc\notin\Span_{\bF_q}(\boldd)$ satisfy that
		$\bolda^\intercal \boldb=\boldc^\intercal \boldd$ if and only if~$\bolda=\lambda \boldc$ and~$\boldb=\lambda^{-1}\boldd$ for some~$\lambda\in\bF_q^*$.
	\end{enumerate}
Therefore,~\ref{item:dependence} and~\ref{item:sym} imply that~$\cQ_k$ contains~$q^k-1$ equivalence classes of size one. In addition,~\ref{item:dependence} and~\ref{item:nonsym} imply that $\cQ_k$ contains~$\frac{(q^k-1)(q^k-q)}{2(q-1)}$ equivalence classes of size two.
\end{proof}

\section*{Appendix~B}\label{appendix:randomizedEncryption}
\textbf{Randomized encryption.} In certain cryptographic scenarios it is imperative that repeated messages will not induce repeated cyphertexts. That is, Eve must not know that Bob wishes to send Alice a message which has already been sent before. A common solution to this constraint is to use randomness, as suggested below.

In this section it is shown that this can be attained at the price of a slightly larger public key and half of the information rate. Roughly speaking, the idea behind the suggested scheme is to modify the Sidon cryptosystem so that one of the elements~$\bolda$ and~$\boldb$ (see Section~\ref{section:TheSidonCryptosystem}) is random. 

To achieve the above goal, Alice randomly chooses an additional irreducible polynomial $P_R(x)\in\bF_q[x]$ of degree~$k=n/2$, which is appended to the public key, and fixes a canonical basis~$z_1,\ldots,z_k$ of~$F_R\triangleq \bF_q[x]\bmod (P_R(x))$ over~$\bF_q$. Rather than encoding into 
the set~$\cQ_k$, Bob encodes his message as~$\bolda=(a_i)_{i=1}^k\in\bF_q^k\setminus\{0\}$.
By randomly picking $\boldb=(b_i)_{i=1}^k\in\bF_q^k\setminus\{0\}$ and using the canonical basis~$z_1,\ldots,z_k$, Bob defines
\begin{align*}
	\hat{a}&\triangleq \sum_{i=1}^{k}a_iz_i\;, & \hat{b}&\triangleq \sum_{i=1}^{k}b_iz_i\;,\\
	\hat{a}/\hat{b}&\triangleq\hat{c}=\sum_{i=1}^{k}c_iz_i \;\mbox{, and }&\boldc&\triangleq(c_1,\ldots,c_k)\;,
\end{align*}
and sends~$E(\boldc,\boldb)$ to Alice. Upon receiving $E(\boldc,\boldb)$, Alice decrypts $E(\boldc,\boldb)$ as in Section~\ref{section:TheSidonCryptosystem}, obtains $ \{\boldc',\boldb'\} \triangleq \{\lambda\cdot(c_1,\ldots,c_k),\tfrac{1}{\lambda}\cdot(b_1,\ldots,b_k)\}$, for some unknown~$\lambda\in\bF_q$, and computes
\begin{align*}
	\left( \lambda\sum_{i=1}^{k}c_iz_i \right)\left( \frac{1}{\lambda}\sum_{i=1}^{k}b_iz_i \right)=\lambda\hat{c}\cdot\tfrac{1}{\lambda}\hat{b}=\hat{a}.
\end{align*}
By representing~$\hat{a}$ over~$z_1,\ldots,z_k$, Alice obtains~$\bolda$, the plaintext by Bob. It is evident from this scheme that repeatedly transmitting a message~$\bolda$ is highly unlikely to produce identical cyphertexts.

\section*{Appendix~C}\label{appendix:LinearizationAttack}
\textbf{Linearization attack over~$\bF_q$.} 
To obtain an~$\bF_{q^n}$-solution to~$\boldOmega_{\text{lin}}$, set~$y_s=\sum_{j\in[n]}y_{s,j}\delta_j$ for every~$s\in[n]$, where~$\{ \delta_i \}_{i\in[n]}$ is any basis of~$\bF_{q^n}$ that the attacker chooses, and the~$y_{i,j}$'s are variables that represent~$\bF_q$ elements. This transition to~$\bF_q$ can be described using Kronecker products as follows. Let~$\{ c_d^{(i,j)} \}_{i,j,d\in[n]}$ so that~$\delta_i\delta_j=\sum_{d\in[n]}c_d^{(i,j)}\delta_d$ for every~$i$ and~$j$. 
Then, an equation of the form~$\sum_{s,t}\alpha_{s,t}y_sy_t=0$, where~$\alpha_{s,t}\in\bF_q$, is written as
\begin{align*}
	0=\sum_{s,t}\alpha_{s,t}\left( \sum_{j\in[n]}y_{s,j}\delta_j \right)\left( \sum_{i\in[n]}y_{t,i}\delta_i \right)&=
	\sum_{s,t}\alpha_{s,t}\sum_{i,j}\left( \sum_{d\in[n]}c_d^{(i,j)} \delta_d \right)y_{s,j}y_{t,i}\nonumber\\
	&= \sum_{d\in[n]}\delta_d \sum_{s,t,i,j}\alpha_{s,t}c_d^{(i,j)}y_{s,j}y_{t,i}.
\end{align*}
Hence, since~$\{\delta_d\}_{d\in[n]}$ is a basis, and since the remaining scalars are in~$\bF_q$, this equation induces the~$n$ equations
\begin{align}\label{equation:FqCoefficients}
	\sum_{s,t,i,j}\alpha_{s,t}c_a^{(i,j)}y_{s,j}y_{t,i}=0 \text{ for every }a\in[n].
\end{align}
To describe the resulting~$\bF_q$-system, recall that the~$\binom{n+1}{2}$ columns of~$\boldOmega_{\text{lin}}$ are indexed by the monomials~$y_sy_t$, and the rows are indexed by minors~$((i,j),(\ell,d))$. The value of an entry in column~$y_sy_t$ and row~$((i,j),(\ell,d))$ is the coefficient of~$y_sy_t$ in the equation for the~$2\times 2$ zero minor of $(\sum_{i\in[n]}y_i\boldM^{(i)})$ in rows~$i$ and~$j$, and columns~$\ell$ and~$d$. 
It follows from~\eqref{equation:FqCoefficients} 
that by performing the transition to~$\bF_q$ we obtain the coefficient matrix~$\boldOmega_q\triangleq\boldOmega_{\text{lin}}\otimes \boldC$, where~$\boldC\in\bF_q^{n\times n^2}$ is a matrix which contains the~$c_d^{(i,j)}$'s\footnote{More precisely, $\boldC$ contains~$c_d^{(i,j)}$ in entry~$(d,(i,j))$, where the~$n^2$ columns are indexed by all pairs~$(i,j)$, $i,j\in[n]$.}, and~$\otimes$ denotes the Kronecker product.

In the resulting system~$\boldOmega_q$ the columns are indexed by~$y_{s,j}y_{t,i}$ with~$s\le t$, and the rows are indexed by~$(((i,j),(\ell,d)),a)$ (i.e., a pair which represents a minor, up to the symmetry in~\eqref{equation:symmetryReduction}, and an index~$a$ from~\eqref{equation:FqCoefficients}). This matrix contains
\begin{align}\label{equation:SizeOfBoldCBoldOmega}
	\binom{n+1}{2}\cdot n^2 &= 8k^4+4k^3\triangleq k'\text{ columns, and}\nonumber\\
	\binom{\binom{k}{2}+1}{2}\cdot n &= \frac{2k^5-4k^4+6k^3-4k^2}{8}\text{ rows,}
\end{align}
and a satisfying assignment for the~$y_{i,j}$'s corresponds to a vector over~$\bF_q$ in its right kernel. As explained earlier, since there exists a solution (the secret key in the Sidon cryptosystem), the dimension of the right kernel of~$\boldOmega_q$ is at least one. Since the system has more rows than columns, one would expect the dimension of the kernel to be at most one, and the solution should be found by standard linearization techniques (see below). However, we have that~
\begin{align}\label{equation:rankKronecker}
	\rank(\boldOmega_q)&=\rank(\boldOmega\otimes \boldC)=\rank(\boldOmega)\cdot\rank(\boldC)\nonumber\\
	&\le \binom{n+1}{2}\cdot n=4k^3+2k^2\ll k',
\end{align}
and therefore, the attacker is left with yet another MinRank problem, where a rank-one solution (in a sense that will be made clear shortly) should be found in the linear span of at least~$4k^4-2k^2$ matrices~$\{\matricize(\boldv_i)\}$, where~$\{\boldv_i\}\subseteq \bF_q^{k'}$ span~$\ker(\boldOmega_q)$.


To describe in what sense a solution of~$\boldOmega_q$ is of ``rank one,'' index the columns of~$\boldOmega_q$ by tuples~$\cR\triangleq \{(s,j,t,i)\vert s,j,t,i\in[n], s\le t\}$, where each column corresponds to the monomial~$y_{s,j}y_{t,i}$. Then, one must finds a vector~$\boldz=(z_{s,j,t,i})_{(s,j,t,i)\in\cR}\in\bF_q^{\binom{n+1}{2}n^2}$ in the right kernel of~$\boldOmega_q$, and wishes to decompose it to find the respective values for~$\{y_{s,j}\}_{s,j\in[n]}$. 

To find that a solution~$\boldz$ is not parasitic, i.e., that it corresponds to some solution~$\hat{\boldy}=(\hat{y}_{s,j})_{s,j\in[n]}$ to the original quadratic system, one must verify that all the following matrices are of rank one, and moreover, that there exist~$\hat{\boldy}$ which satisfies the following rank one decompositions.

\begin{align*}
	&\begin{pmatrix}
		z_{1,1,1,1} & \ldots & z_{1,1,1,n} & z_{1,1,2,1} & \ldots & z_{1,1,2,n} & \ldots & z_{1,1,n,1} & \ldots & z_{1,1,n,n}\\
		z_{1,2,1,1} & \ldots & z_{1,2,1,n} & z_{1,2,2,1} & \ldots & z_{1,2,2,n} & \ldots & z_{1,2,n,1} & \ldots & z_{1,2,n,n}\\
		\vdots & \ddots & \vdots & \vdots & \ddots & \vdots & \ddots & \vdots & \ddots & \vdots\\
		z_{1,n,1,1} & \ldots & z_{1,n,1,n} & z_{1,n,2,1} & \ldots & z_{1,n,2,n} & \ldots & z_{1,n,n,1} & \ldots & z_{1,n,n,n}\\
	\end{pmatrix}=(\hat{y}_{1,1},\ldots,\hat{y}_{1,n})^\intercal\hat{\boldy}\\
	&\begin{pmatrix}
	z_{2,1,2,1} & \ldots & z_{2,1,2,n} & \ldots & z_{2,1,n,1} & \ldots & z_{2,1,n,n}\\
	z_{2,2,2,1} & \ldots & z_{2,2,2,n}  & \ldots & z_{2,2,n,1} & \ldots & z_{2,2,n,n}\\
	\vdots & \ddots & \vdots & \ddots & \vdots & \ddots & \vdots\\
	z_{2,n,2,1} & \ldots & z_{2,n,2,n} & \ldots & z_{2,n,n,1} & \ldots & z_{2,n,n,n}\\
\end{pmatrix}=(\hat{y}_{2,1},\ldots,\hat{y}_{2,n})^\intercal(\hat{y}_{2,1}\ldots,\hat{y}_{n,n})\\
&\qquad\vdots\\
	&\begin{pmatrix}
	z_{n,1,n,1} & \ldots & z_{n,1,n,n} \\
	z_{n,2,n,1} & \ldots & z_{n,2,n,n} \\
	\vdots & \ddots & \vdots\\
	z_{n,n,n,1} & \ldots & z_{n,n,n,n}\\
\end{pmatrix}=(\hat{y}_{n,1},\ldots,\hat{y}_{n,n})^\intercal(\hat{y}_{n,1}\ldots,\hat{y}_{n,n})\\
\end{align*}
It is readily verified that~$\boldz$ is non-parasitic if and only if all the above are met. However, it is difficult to find such a solution, since the dimension of the right kernel of~$\boldOmega_q$ is at least 
\begin{align*}
	\text{num. of columns}-\rank(\boldOmega_q)&= k'-\rank(\boldOmega_q)\\
	&\ge 8k^4+4k^3-\left( \binom{n+1}{2}-n \right)n=8k^4+2k^2.
\end{align*}

\section*{Appendix~D}\label{appendix:difBasisExt}
\textbf{Different basis extension.} In the spirit of Remark~\ref{remark:extensionsDontWork}, when trying to find different extensions of~$\nu_1,\ldots,\nu_k$ to a basis of~$\bF_{q^n}$ such that the resulting~$\boldGamma$ and~$\boldGamma'$ have disjoint kernels, we discover the following. For~$\boldnu=(\nu_1,\ldots,\nu_k)$, let $\boldu_1=(\boldnu,\boldmu_1)$, $\boldu_2=(\boldnu,\boldmu_2)$ be two possible extensions, with respective matrices~$\boldE_1,\boldE_2$ (i.e., such that~$\boldbeta=\boldu_1\boldE_1=\boldu_2\boldE_2$), respective~$\{ \boldN^{(i)} \}_{i\in[n]}$, $\{ \tilde{\boldN}^{(i)} \}_{i\in[n]}$, and respective  kernel vectors~$\{ \boldn_i \}_{i\in[n]}$ and~$\{ \tilde{\boldn}_i \}_{i\in[n]}$. Let~$\boldbeta^\intercal\boldbeta=\sum_{i\in[n]}\beta_i\boldB^{(i)}$, 
and observe that
\begin{align*}
	\boldu_1^\intercal\boldu_1&=(\boldE_1^{-1})^\intercal\boldbeta^\intercal\boldbeta\boldE_1^{-1}=\sum_{i\in[n]}\beta_i (\boldE_1^{-1})^\intercal \boldB^{(i)}\boldE_1^{-1}\\
	&=\sum_{i\in[n]}\beta_i \boldN^{(i)},\text{ and therefore }\boldE_1^\intercal \boldN^{(i)} \boldE_1=\boldB^{(i)}.\text{ Similarly,}\\
	\boldu_2^\intercal\boldu_2&=(\boldE_2^{-1})^\intercal\boldbeta^\intercal\boldbeta\boldE_2^{-1}=\sum_{i\in[n]}\beta_i (\boldE_2^{-1})^\intercal \boldB^{(i)}\boldE_2^{-1}\\
	&=\sum_{i\in[n]}\beta_i \tilde{\boldN}^{(i)}\text{ and therefore }\boldE_2^\intercal \tilde{\boldN}^{(i)} \boldE_2=\boldB^{(i)}.
\end{align*}
Hence, the respective kernel vectors~$\{ \boldn_i \}_{i\in[n]}$ and~$\{ \tilde{\boldn}_i \}_{i\in[n]}$  are \textit{identical}. 
\fi
\end{document}